\documentclass[twoside]{llncs}
\usepackage{fullpage}
\usepackage{amsmath,amssymb}
\PassOptionsToPackage{normalem}{ulem}\usepackage{ulem}

\usepackage{enumitem}\setlist{nosep}

\usepackage{xfrac}
\usepackage{varwidth}
\usepackage{linegoal}

\usepackage{tikz}
\usepackage{xspace}
\newcommand{\Agents}{\mathcal{N}}%
\newcommand{\Alternatives}{\mathcal{A}}%
\newcommand{\naive}{na\"ive\xspace}
\newcommand{\OW}{\text{O/W}}%

\newcommand{\xPrime}[1]{#1^{\prime}}%

\newcommand{\st}{\text{ s.t. }}%
\newcommand{\SetSt}{\;\middle\vert\;}%
\newcommand{\sizeof}[1]{\left|#1\right|}%
\newcommand{\argmax}{\mathop{\mathrm{argmax}}}%

\newcommand{\OneBB}[1]{\mathds{1}}%

\newcommand{\RobertsPhi}[1]{\varphi_{#1}}%
\newcommand{\RPhiOf}[2]{\RobertsPhi{#1}\left(#2\right)}%

\newcommand{\UtilDom}{\mathcal{U}}%
\newcommand{\TypeDom}{\mathcal{T}}%
\newcommand{\UtilDomI}[1]{\UtilDom_{#1}}%
\newcommand{\TypeDomI}[1]{\TypeDom_{#1}}%
\newcommand{\UtilDomQL}{\UtilDom^{QL}}%

\newcommand{\prefBy}{\preccurlyeq}%
\newcommand{\prefOver}{\succcurlyeq}%
\newcommand{\strictPrefBy}{\prec}%
\newcommand{\strictPrefOver}{\succ}%
\newcommand{\Ind}{\sim}%
\newcommand{\WTP}[1]{p^{#1}}%

\newcommand{\MECH}[2]{\left\langle #1,#2\right\rangle }%
\newcommand{\OUTCOME}[2]{\left\langle #1,#2\right\rangle }%
\newcommand{\UtilTemplate}[4]{#1_{#2}\left(#3,#4\right)}%
\newcommand{\xOf}[1]{x\left(#1\right)}%
\newcommand{\piOf}[2]{p_{#1}\left(#2\right)}%
\newcommand{\AltCost}[1]{c_{#1}}%
\newcommand{\AgentWeight}[1]{w_{#1}}%

\newcommand{\Representation}{representation}%
\newcommand{\Representations}{representations}%
\newcommand{\POSRepresentation}{pos-representation}%
\newcommand{\POSRepresentations}{pos-representations}%
\newcommand{\Represent}{represent}%
\newcommand{\Represents}{represents}%
\newcommand{\POSRepresent}{pos-represent}%
\newcommand{\POSRepresents}{pos-represents}%
\newcommand{\Represented}{represented}%
\newcommand{\POSRepresented}{pos-represented}%

\newcommand{\SRepresents}{\SWITCH{\POSRepresents}{\Represents}}%
\spnewtheorem{claim}{Claim}{\bfseries}{\itshape}
\spnewtheorem{exampleParallelDomain}{Example}{\bfseries}{\rmfamily}
	
\spnewtheorem*{theorem*}{\theoremname}{\bfseries}{\itshape}
\spnewtheorem{mechanism}{Mechanism}{\itshape}{\rmfamily}
\spnewtheorem*{mechanism*}{Mechanism}{\itshape}{\rmfamily}
\spnewtheorem*{claim*}{\claimname}{\bfseries}{\rmfamily}
\spnewtheorem*{corollary*}{\corollaryname}{\bfseries}{\rmfamily}

\newcommand{\SWITCH}[2]{\undefined}
\newcommand{\ThmFootnote}[1]{\hspace{-6pt}\footnote{#1}}

\newcommand{\Nth}[1]{\ifmmode\ERROR\else#1\textsuperscript{\underline{th}}\fi\xspace}

\begin{document}

\def\footnotemark{\relax}
\title{Almost Quasi-linear Utilities in Disguise: Positive-representation\\
	\qquad{}An Extension of Roberts' Theorem%
		\thanks{We would like to thank Reshef Meir and Hongyao Ma for stimulating early discussions on the topic.
		We also would like to thank the anonymous reviewers for their detailed reviews, which helped us to improve the presentation of this work.}%
		\thanks{This work was supported in part by Israel Science Foundation (ISF) Grant 1626/18.}
		\thanks{An extended abstract of this work is forthcoming in:
			The \Nth{15} Conference on Web and Internet Economics (WINE 2019).}}
\author{Ilan Nehama}
\institute{Bar-Ilan University\\\email{ilan.nehama@mail.huji.ac.il}}
\maketitle
\begin{abstract}
This work deals with the implementation of social choice rules using dominant strategies for unrestricted preferences. The seminal Gibbard-Satterthwaite theorem shows that only few unappealing social choice rules can be implemented unless we assume some restrictions on the preferences or allow monetary transfers. When monetary transfers are allowed and quasi-linear utilities w.r.t. money are assumed, Vickrey-Clarke-Groves (VCG) mechanisms were shown to implement any affine-maximizer, and by the work of Roberts, only affine-maximizers can be implemented whenever the type sets of the agents are rich enough.

In this work, we generalize these results and define a new class of preferences: Preferences which are \emph{positive-represented by a quasi-linear utility}. That is, agents whose preference on a subspace of the outcomes, which is defined by a threshold, can be modeled using a quasi-linear utility. We show that the characterization of VCG mechanisms as the incentive-compatible mechanisms extends naturally to this domain.
We show that the original characterization of VCG mechanism is an immediate corollary of our generalized characterization. 
Our result follows from a simple reduction to the characterization of VCG mechanisms. Hence, we see our result more as a fuller more correct version of the VCG characterization than a new non-quasi-linear domain extension.

This work also highlights a common misconception in the community attributing the VCG result to the usage of transferable utility.
Our result shows that these results extend naturally to the non-transferable utility domain. 
That is, that the incentive-compatibility of the VCG mechanisms does not rely on money being a common denominator, but rather on the ability of the designer to fine the agents on a continuous (maybe agent-specific) scale.

We also provide simple characterizations of the types which are \Represented{} and \POSRepresented{}  by quasi-linear utility functions.
We show characterizations both in utility function terms and in preference terms, by that supplying a full comparison of the different classes.

We think these two insights, considering the utility as a representation and not as the preference itself (which is common in the economic community) and considering utilities which represent the preference only for the relevant domain, would turn out to fruitful in other domains as well.

\end{abstract}

\section{Introduction\label{sec:Introduction}}
Consider the problem of a designer who wishes to implement a given social choice rule. That is, consider a finite set of agents $\Agents$ and a finite set of possible social alternatives $\Alternatives$ s.t. each agent holds a preference (a total order) over $\Alternatives$; The designer needs to choose, as a function of the agents' preferences, one social alternative out of $\Alternatives$, e.g., one that maximizes a social welfare or some other criterion, but the preferences are a-priori unknown to the designer. A \naive procedure for the designer would be to first query the preferences of the agents and choose the alternative accordingly. The problem with such solution is that an agent might report her preference untruthfully if she thinks it might result in a better alternative for her, by that preventing the designer from choosing correctly.

In this work, we study the fundamental question of which social choice rules are implementable by a principal, using the most basic implementation concept of dominant-strategies. That is, we aim to characterize which social choice rules can be implemented s.t. no agent has an incentive to misreport her preference, independent of the reports of the other agents, and to show a mechanism for any such implementable social choice rule.
Furthermore, we analyze implementation using deterministic incentive-compatible direct revelation mechanisms. In a \emph{direct revelation mechanism} each agent is asked to report her true preference and subsequently a decision is made according to a decision rule. Incentive-compatibility for these mechanism means that an agent cannot benefit from reporting a preference different from her true preference. In the sequel, we assume the mechanisms are direct and deterministic and in Section~\ref{sec:Discussion} we discuss these restrictions and how to relax them. In particular, note that the assumption of the mechanism being direct is without loss of the generality, as according to the revelation principle~\cite{Myerson79} any implementable social choice rule is implementable using a direct revelation mechanism and the characterization of direct revelation implementations is easily lifted to the general characterization.

The seminal works of Gibbard~\cite{Gibbard1973} and Satterthwaite~\cite{Satterthwaite1975} show that without further assumptions, if all profiles of preferences are feasible, one cannot devise a mechanism for choosing a single alternative, besides trivial procedures which a-priori ignore most agents or most alternatives. Assuming the preferences of the agents are richer and include intensity of preferences (that is, a cardinal utility function $u\colon\Alternatives\rightarrow\Re$ which assigns for every alternative its utility or desirability for the agent) does not circumvent this impossibility.

On the other hand, the seminal works of Vickrey~\cite{Vickrey1960}, Clarke~\cite{clarke1971}, and Groves~\cite{groves1973} show that allowing the designer to induce monetary transfers (pay money to the agents or charge them) gives rise to a non-trivial family of implementable social choice rules, while still assuming that all profile preferences are feasible. That is, the outcome is both a chosen alternative and in addition a sequence of payments of the agents (which might be negative, i.e., payment to the agents). In addition, the preferences of the agents is extended to preferences over $\Alternatives\times\Re$, the tuples of an alternative and a payment of the agent. Note that (under some mild conditions) such preference also defines a monetary value for each alternative and hence defines a cardinal utility $v\colon\Alternatives\rightarrow\Re$ over the alternatives.

When all agents' preferences are defined by quasi-linear utility functions $u\colon\Alternatives\times\Re\rightarrow\Re$ of the form $\UtilTemplate u{}az=v\left(a\right)-z$ for some $v\colon\Alternatives\rightarrow\Re$, the Vickrey-Clarke-Groves~(VCG) mechanisms are incentive-compatible direct mechanisms which can implement any affine-maximizer of the agents' valuation functions~$v$~\cite{Vickrey1960,clarke1971,groves1973}. Moreover, Green and Laffont~\cite{RePEc:ecm:emetrp:v:45:y:1977:i:2:p:427-38} and Roberts~\cite{Roberts1979} later showed that these are the only incentive-compatible onto mechanisms when all profiles of quasi-linear utility functions are feasible.

However, quasi-linearity is a strong assumption. On the individual level, it is an assumption that in particular there is no income effect in the evaluations of the alternatives, which might be violated due to, for example, lack of liquidity, wealth effect, or risk-aversion. On the society level, it is an assumption of a \emph{transferable utility} model, that is, an assumption that the `utility loss' due to paying $\$1$ for Agent~$i$ equals to the `utility gain' due to receiving this amount for Agent~$j$.

Hence, it is of interest to search of other preference domains for which non-trivial implementable social choice rules exist, while still inducing the full unrestricted set of preferences over the alternatives. To the best of our knowledge, except the aforementioned works for quasi-linear utilities, very few papers considered social choice mechanisms with monetary transfers without restricting the possible profiles of preferences over the alternatives (In much more restricted domains, like single-peaked facility location~\cite{Moulin1980} or auctions, just to name two, there are many truthful mechanisms that are not affine-maximizers.).

The work most relevant to this paper is the result by Ma et al.~\cite{Ma2018}.
In this work, Ma et al. present a new family of utility functions, the \emph{parallel utility functions}, which is a super-set of the family of the quasi-linear utility functions. Following the works of Vickrey, Clarke, and Groves and the work of Roberts, they present a family of mechanisms which are incentive-compatible whenever all agents have parallel utility functions, and prove that the induced social choice rules are the only implementable rules (under two additional mild conditions, \emph{No subsidy} \& \emph{Individual rationality}) if all induced valuation profiles are feasible.

\subsection*{Our Results}

We think that parallel utility functions can be seen as a special case of disguised almost quasi-linear utility functions, by that strengthening the knife-edginess of the quasi-linear scenario and the characterization of VCG mechanisms.
We formalize this by two insights (which catches this intuitive notion of `disguised almost quasi-linear'). We discuss this connection we see in more details in Section~\ref{sec:Discussion}.

First, we notice that VCG mechanisms do not depend on the quasi-linear \Representation{} an agent reports and can be stated (while still maintaining the characterization result) as mechanisms that get as input the preferences of the agents and not utility functions. In particular, VCG mechanisms can be extended to receive non-quasi-linear utility functions as long as there is some quasi-linear \Representation{} of it (quasi-linear in disguise). While this observation is technically trivial, we think it is conceptually important as being consistent with the \emph{revealed preferences principle} and in particular with the fact that one cannot test whether an agent holds one utility \Representation{}or the other, so in general there is no base to assume that one utility function better \Represents{} the agent.

Second, our main insight is a weaker notion of \Representation{}, \emph{positive-\Representation{}} (and shortly \emph{\POSRepresentation{}}). We say an agent's preference is \POSRepresented{} by a utility function if they coincide whenever the utility of at least one of the two comparands is positive (and not for any two comparands as with regular \Representation{} by a utility function). This new simple notion allows us to extend the previous above works of Vickrey, Clarke, Groves, Green-Laffont, Roberts.

We extend the VCG mechanisms to mechanisms that get as input preferences which are \POSRepresented{} by quasi-linear functions. We show that these mechanisms are incentive compatible~(Theorem~\ref{thm:MyVCG}) and show that when all induced valuation profiles are feasible, essentially these are the only incentive-compatible onto mechanisms~(Theorem~\ref{thm:MyRoberts}). This result extends the characterizations of Roberts~\cite{Roberts1979} and Green-Laffont~\cite{RePEc:ecm:emetrp:v:45:y:1977:i:2:p:427-38}.

We argue that preferences which are \POSRepresented{} by quasi-linear utility functions maintain the essential properties of quasi-linear utility functions (hence we call them `disguised almost quasi-linear utility functions') and we see the simple proofs by reduction as a supporting evidence for that. This strengthens the knife-edginess of the VCG phenomenon by showing that the quasi-linear domain is essentially the only unrestricted preference domain for which an incentive-compatible non-trivial implementation is known, and we hope this gives a better research direction of exploring mechanism design beyond quasi-linear preferences.

Furthermore, we think that our work addresses a common misconception in the community of the reason why introducing money enables the design of general mechanisms and circumventing impossibility results \`a la Gibbard-Satterthwaite. Many works attribute the VCG result to the extended model being a \emph{transferable utility} model. For example, in the \emph{Algorithmic Game Theory} book~\cite{AGTbookCh9}, we found that VCG Mechanisms (and the driving force behind their incentive compatibility) are introduced using phrases like ``Money can be transferred between players. ... (This) will allow us to do things that we could not do otherwise.'' Our work shows that the main driving-force is the individual additive-separability and that the introduction of transferable utility does not play a role in the characterization. For instance, the principal might use different incomparable payment methods for different agents. We discuss the connection to the stronger assumption of quasi-linearity and transferable utility in more details in Section~\ref{sec:Discussion}.

Identifying the preference of an agent with her cardinal utility is basic in the theory of modeling economic agents, and in most scenarios, unless additional strong assumptions are added, the utility is only a representation of the preference and does not hold any additional information of the decision behavior of the agent. Moreover, in many cases, one's assumptions on the agents' preferences are mostly irrelevant when considering outcomes below some threshold or an outside option. Hence, we think these two insights could be fruitful in other domains as well.

To ease the reading flow, some of the less significant proofs are postponed to the appendix.

\section{Model}
We denote the set of agents by $\Agents$, the set of alternatives by $\Alternatives$, and their respective size by $n=\sizeof{\Agents}$ and $m=\sizeof{\Alternatives}$. Each agent holds a \emph{preference}, i.e. an order, over $\Alternatives\times\Re$, the set of tuples of an alternative and a payment of the agent. We also refer to the preference as the agent's \emph{type. }We say a cardinal function $u\colon\Alternatives\times\Re\rightarrow\Re$ \emph{\Represents{}}~\cite[Def.~1.B.2]{MasColell1995} a preference $\prefOver$ (and refer to $u$ as the \emph{utility function}) if for any two alternatives $a,b\in\Alternatives$ and two payments $z_{a},z_{b}\in\Re$,
	\[ \OUTCOME a{z_{a}} \prefOver\OUTCOME b{z_{b}} \,\iff\,\UtilTemplate u{}a{z_{a}}\geqslant\UtilTemplate u{}b{z_{b}}\text{.} \]

It is not hard to see that two utility functions $u,w\colon\Alternatives\times\Re\rightarrow\Re$ \Represent{} the same preference iff one is a monotone transformation of the other, that is, there exists a monotone bijection $\varphi\colon\Re\rightarrow\Re$ s.t. $u=\varphi\circ w$. In this work, we assume that the preferences of the agents can be \Represented{} by a utility function $u$ s.t. $u$ is strictly decreasing in the payment, i.e., an agent prefers to pay less. We denote the set of all utility functions which are strictly decreasing in the payment by $\UtilDom$ and the set of preferences (types) \Represented{} by utility functions from $\UtilDom$ by $\TypeDom$. We use the notations of types and type sets in our results to emphasize that our results are invariant to the choice of \Representations{}.
In places it is not important, in order to ease the readability, we sometimes use utility functions to characterize directly the agents.

A special family of utility functions are the \emph{quasi-linear} utility functions.\footnote
	{In the economics literature (e.g.,~\cite{RePEc:ecm:emetrp:v:45:y:1977:i:2:p:427-38}) these functions are also referred to as \emph{separable} or \emph{additively separable}.}
These are the utility functions of the form $\UtilTemplate u{}az=v\left(a\right)-z$ for some function $v\colon\Alternatives\rightarrow\Re$ which is referred to as the \emph{valuation}. We denote the set of all quasi-linear utility functions by $\UtilDomQL$. We note that the quasi-linear \Representations{} of a preference are closed to shifts by a constant and in particular a preference can be \Represented{} by a continuum of quasi-linear utility functions.

\begin{claim}\label{claim:RepIsUniqueIFF}
	Let $u,w\colon\Alternatives\times\Re\rightarrow\Re$ be two quasi-linear utility functions. Then the following three statements are equivalent:
\begin{enumerate}
\item Both $u$ and $w$ \Represent{}  the same preference over $\Alternatives\times\Re$.
\item There exists a constant $C\in\Re$ s.t. for any alternative $a\in\Alternatives$ \quad{} $\UtilTemplate u{}a0=\UtilTemplate w{}a0+C$.
\item There exists a constant $C\in\Re$ s.t. for any alternative $a\in\Alternatives$ and $z\in\Re$ \quad{} $\UtilTemplate u{}az=\UtilTemplate w{}az+C$.
\end{enumerate}
\end{claim}

In a recent work, Ma el al. defined a new class of utility function extending the quasi-linear family, which they named \emph{Parallel utility functions}.
\begin{definition}[Parallel utility functions~\cite{Ma2018}]~\par\nopagebreak\ignorespaces%
\def\UtilOfWorstAlt{\min_{\xPrime a\in\Alternatives}\UtilTemplate u{}{\xPrime a}0}%
 A utility function $u\colon\Alternatives\times\Re\rightarrow\Re$ is a parallel utility function if
	\textbullet{} $\UtilTemplate u{}az$ is continuous and strictly decreasing in $z$,
	\textbullet{} For any alternative $a\in\Alternatives$, $\lim_{z\rightarrow\infty}\UtilTemplate u{}az<\UtilOfWorstAlt$, i.e., an agent prefers her worst alternative for free over paying a large enough payment for $a$, and
	\textbullet{} For any two alternatives $a,b\in\Alternatives$ s.t. $\UtilTemplate u{}a0\geqslant\UtilTemplate u{}b0$ and any payment $z\geqslant0$ s.t. $\UtilTemplate u{}bz\geqslant\UtilOfWorstAlt$ (i.e., any positive payment $z$ s.t. the agent prefers paying $z$ for $b$ over receiving some alternative for free)
		\[ \UtilTemplate u{}a{z+\left(\WTP{a}-\WTP{b}\right)}=\UtilTemplate u{}bz\text{,} \]
	for $\WTP{a}$ being the payment $z$ s.t. $\UtilTemplate u{}az=\UtilOfWorstAlt$, i.e., $\WTP{a}$ is the maximal payment s.t. the agent prefers paying it for $a$ over receiving for free her worst alternative, and $\WTP{b}$ is defined similarly.%
	\def\UtilOfWorstAlt{\min_{\xPrime a\in\Alternatives}\UtilTemplate u{}{\xPrime a}0\undefined}%
\end{definition}

\begin{exampleParallelDomain}
Let $u\in\UtilDomQL$ be a quasi linear utility function, i.e., $\UtilTemplate u{}az=v\left(a\right)-z$ for some valuation function $v\colon\Alternatives\rightarrow\Re$.
Then, for any two alternatives $a,b\in\Alternatives$, $\WTP{a}=v\left(a\right)-\min_{\xPrime a\in\Alternatives}v\left(\xPrime a\right)$ and
	\[ \UtilTemplate u{}a{z+\left(\WTP{a}-\WTP{b}\right)}
		\,\,\,\,= v\left(a\right)-\left[z+\left(v\left(a\right)-v\left(b\right)\right)\right]
		\,\,\,\,= v\left(b\right)-z
		\,\,\,\,= \UtilTemplate u{}bz\text{,} \]
i.e., $u$ is a parallel utility function.
\end{exampleParallelDomain}

\begin{exampleParallelDomain}
Generalizing the former, let $u\in\UtilDomQL$ be a quasi-linear utility function and consider a utility function $\xPrime u\in\UtilDom$ which coincides with $u$ whenever $\UtilTemplate u{}az\geqslant\min_{\xPrime a\in\Alternatives}\UtilTemplate u{}{\xPrime a}0$ (but still $\xPrime u$ is continuous and downward monotone in the payment). For any two alternatives $a,b\in\Alternatives$ s.t. $\UtilTemplate{\xPrime u}{}a0\geqslant\UtilTemplate{\xPrime u}{}b0$ and any payment $z\geqslant0$ s.t. $\UtilTemplate{\xPrime u}{}bz\geqslant\min_{\xPrime a\in\Alternatives}\UtilTemplate{\xPrime u}{}{\xPrime a}0$, it holds that \textbullet{} $\UtilTemplate u{}a0=\UtilTemplate{\xPrime u}{}a0\geqslant\UtilTemplate{\xPrime u}{}b0=\UtilTemplate u{}b0$, \textbullet{} $\WTP{a}=\UtilTemplate u{}a0-\min_{\xPrime a\in\Alternatives}\UtilTemplate u{}{\xPrime a}0$
and similarly for $b$. So $\WTP{a}\geqslant \WTP{b}$, and hence
	\[ \UtilTemplate{\xPrime u}{}bz\,\,\,\,=\,\,\,\,\UtilTemplate u{}bz\,\,\,\,=\,\,\,\,\UtilTemplate u{}a{z+\left(\WTP{a}-\WTP{b}\right)}\,\,\,\,=\,\,\,\,\UtilTemplate{\xPrime u}{}a{z+\left(\WTP{a}-\WTP{b}\right)}\text{,} \]
i.e., $u$ is a parallel utility function.
\end{exampleParallelDomain}

\begin{exampleParallelDomain}\label{examplr:PosRepButNotParallel}
	On the other hand, consider the following utility function.
		\[ \UtilTemplate u{}az=\begin{cases} z\leqslant-3 & -3-z\\ z\geqslant-3 & -1-\sfrac z3 \end{cases}
		\quad;\quad \UtilTemplate u{}bz=-2-z \quad;\quad
		\UtilTemplate u{}cz=\begin{cases} z\leqslant-1 & -1-z\\ z\geqslant-1 & -3-3z\text{.} \end{cases}\]
Then, $u$ coincides with the quasi-linear utility function $\left\{ \begin{array}{l}
	\UtilTemplate{u^{QL}}{}az=-3-z\\
	\UtilTemplate{u^{QL}}{}bz=-2-z\\
	\UtilTemplate{u^{QL}}{}cz=-1-z
\end{array}\right.$ whenever $\UtilTemplate{u^{QL}}{}xz\geqslant0=3+\min_{\xPrime a\in\Alternatives}\UtilTemplate{u^{QL}}{}{\xPrime a}0$ but it is not a parallel utility function. One could see that by noticing that $\WTP{a}=6$, $\WTP{b}=1$, and $\WTP{c}=0$, but for $z=0$:
	\[\begin{array}{l}
	\UtilTemplate u{}a{z+\left(\WTP{a}-\WTP{b}\right)}=\UtilTemplate u{}a5=-\sfrac 83\\
	\UtilTemplate u{}bz=\UtilTemplate u{}b0=-2\text{.}
	\end{array}\]
\end{exampleParallelDomain}

Given type domains $\TypeDomI 1,\ldots,\TypeDomI n\subseteq\TypeDom$ for the agents, a \emph{direct mechanism with monetary transfers} (or shortly a \emph{mechanism}) for $\times_{i\in\Agents}\TypeDomI i$ is defined by two functions $x\colon\times_{i\in\Agents}\TypeDomI i\rightarrow\Alternatives$ and $p\colon\times_{i\in\Agents}\TypeDomI i\rightarrow\Re^{\Agents}$ (and we notate by $p_{i}$ the \Nth{i} coordinate of $p$) in the following way: At the first stage each agent reports her type, and based on the reports $t_{1},\ldots,t_{n}$, the alternative $\xOf{t_{1},\ldots,t_{n}}$ is chosen and Agent~$i$ pays $\piOf i{t_{1},\ldots,t_{n}}$. We say that a mechanism $\MECH xp $ is \emph{onto} if the function $x$ is onto, i.e., any alternative is the outcome of some report vector.\footnote{In the social choice and economics literature~\cite{Aleskerov200295,moulin1995cooperative}, this property is also referred to as \emph{Citizen sovereignty} and as \emph{Non-imposition}.}
We use the notation $\TypeDomI{-i}$ for $\times_{j\neq i}\TypeDomI j$, and say that a mechanism $\MECH xp $ is \emph{incentive compatible} if it is in an agent best interest to always report truthfully her type. Formally, for any type $t_{i}\in\TypeDomI i$ and any sequence of reports $t_{-i}\in\TypeDomI{-i}$,
	\[ t_{i}\in\argmax_{\xPrime t_{i}\in\TypeDomI i}\OUTCOME {\xOf{\xPrime t_{i},t_{-i}}}{\piOf i{\xPrime t_{i},t_{-i}}} \text{,} \]
where the maximum is taken according to the preference $t_{i}$.

An interesting class of incentive-compatible mechanisms are the \emph{affine VCG mechanisms} which get as reports utility functions and not general preferences.
\begin{definition}[Affine VCG mechanism]~\par\nopagebreak\ignorespaces%
 We say that a mechanism $\MECH xp $ is an \emph{affine VCG mechanism} if there exist
	an agent weight vector $w\in\Delta\left(\Agents\right)$ (i.e., $w\in\left[0,1\right]^{\Agents}$ and $\sum_{i\in\Agents}w_{i}=1$),
	an alternative cost vector $c\in\Re^{\Alternatives}$,
	and a non-empty set of alternatives $\xPrime{\Alternatives}\subseteq\Alternatives$
s.t. for any report vector $\left(u_{1},\ldots,u_{n}\right)\in\times_{i\in\Agents}\UtilDomI i$,
	\[ \xOf{u_{1},\ldots,u_{n}}\in\argmax_{a\in\xPrime{\Alternatives}}\left(c_{a}+\sum_{i\in\Agents}w_{i}\cdot\UtilTemplate uia0\right) \]
and there exist functions $h_{i}\colon\UtilDomI{-i}\rightarrow\Re$ for $i=1,\ldots,n$ s.t. for any report vector $\left(u_{1},\ldots,u_{n}\right)$ the payment of Agent~$i$ is
	\[ \piOf i{u_{i},u_{-i}}=h_{i}\left(u_{-i}\right)-\sfrac 1{\AgentWeight i}\cdot\left(\sum_{j\neq i}w_{j}\cdot\UtilTemplate uj{a^{\star}}0+c_{a^{\star}}\right)\qquad\text{for }a^{\star}=\xOf{u_{i},u_{-i}} \]
if $\AgentWeight i>0$ and $\piOf i{u_{i},u_{-i}}=h_{i}\left(u_{-i}\right)$ if $\AgentWeight i=0$ (i.e., if Agent~$i$ has no influence of the chosen alternative $\xOf{u_{i},u_{-i}}$).
\end{definition}

In their works, Vickrey~\cite{Vickrey1960}, Clarke~\cite{clarke1971}, and Groves~\cite{groves1973} proved that when for all agents $\UtilDomI i\subseteq\UtilDomQL$, i.e., all agents hold quasi-linear utility functions, these mechanisms are incentive-compatible. In particular, an agent is indifferent between reporting different representations of her type. Furthermore, recalling that two (quasi-linear) utility functions represent the same preference iff they differ by a constant, we note that $x$ and $p_{i}$ are actually invariant to the representation an Agent~$i$ reports.
Later, Green and Laffont~\cite{RePEc:ecm:emetrp:v:45:y:1977:i:2:p:427-38} and Roberts~\cite{Roberts1979} proved that if for all agents $\UtilDomI i=\UtilDomQL$, these are the only direct onto incentive-compatible mechanisms.
\begin{theorem*}[{Roberts~\cite[Thm.~3.1]{Roberts1979} \& Green and Laffont~\cite[Thm.~3]{RePEc:ecm:emetrp:v:45:y:1977:i:2:p:427-38}}]~\par\nopagebreak\ignorespaces%
Let $\MECH xp $ be an onto incentive-compatible mechanism for the case $\left[\forall i\quad\UtilDomI i=\UtilDomQL\right]$. Then, it is an affine VCG mechanism.
\end{theorem*}

\section{Main Result}
In this work, we introduce a weaker notion of \Representation{} we call positive-representation.
\begin{definition}[Positive-representation]~\par\nopagebreak\ignorespaces%
We say a utility function $u\colon\Alternatives\times\Re\rightarrow\Re$ \emph{\POSRepresents{} } a preference $\prefOver$ if for any two alternatives $a,b\in\Alternatives$ and payments $z_{a},z_{b}\in\Re$ s.t. either $\UtilTemplate u{}a{z_{a}}\geqslant0$ or $\UtilTemplate u{}b{z_{b}}\geqslant0$,
	\[ \MECH a{z_{a}}\prefOver\MECH b{z_{b}}\quad\iff\quad\UtilTemplate u{}a{z_{a}}\geqslant\UtilTemplate u{}b{z_{b}}\text{.} \]
\end{definition}

That is, we require $u$ to \Represent{} the preference only above some threshold (The choice of zero as the threshold is for normalization). E.g., in Example~\ref{examplr:PosRepButNotParallel} we saw a preference which coincides with a quasi-linear utility function $u^{QL}$ whenever the utility is positive (but the threshold was not $\min_{a\in\Alternatives}\UtilTemplate{u^{QL}}{}a0$ and the utility was not a parallel utility). One could think of this threshold as representing an outside option so we are not interested in modeling the agent's preferences between two unacceptable alternatives.
This allows one to \Represent{} the agent using a simpler utility function (quasi-linear in our case). For example, a common practice is to model a buyer in a combinatorial auction using a valuation function $v\colon2^{X}\rightarrow\Re$ which assigns a monetary value to each subset of items, and a quasi-linear utility function $\UtilTemplate u{}Sz =v\left(S\right)-z$ for $S\subseteq X$ and $z\in\Re$. In particular, such \Representation{} entails a preference between two outcomes in which the buyer over-paid for the bundle she received. But in most cases, it is assumed a buyer can always refuse to over-pay, and it might be a better modeling to assume indifference. Moreover, commonly this preference pays no role in the analysis, so the analysis should be applicable also if the buyer has a different preference between such outcomes.

In this work we are interested in agents that can be \POSRepresented{} by quasi-linear utility functions, and show that the characterization of incentive-compatible mechanisms is naturally extended for these preferences.

It is straight-forward from the definition that if a utility function $u$ \Represents{} a preference then in particular $u$ \POSRepresents{} the preference. Moreover, as a corollary of Claim~\ref{claim:RepIsUniqueIFF} we get the following.
\begin{corollary}
Let $u\in\UtilDomQL$ be a quasi-linear utility function and $\prefOver$ a preference s.t. $u$ \Represents{} $\prefOver$. Then, for any constant $C\in\Re$ the utility function $u+C$ \POSRepresents{} $\prefOver$.
\end{corollary}

Note that only a weaker version of Claim~\ref{claim:RepIsUniqueIFF} holds for \POSRepresentations{}. One can find two quasi-linear utility functions $u,v\in\UtilDomQL$ which differ by a constant but $u$ \POSRepresents{} a preference which $v$ does not. For instance, take $\Alternatives=\left\{ a,b,c\right\} $ and the preference which is defined by the following utility function $w\colon\Alternatives\times\Re\rightarrow\Re$:
	\[ \UtilTemplate w{}az=1-z\quad;\quad\UtilTemplate w{}bz=\begin{cases}
		z\leqslant2 & 2-z\\
		z\geqslant2 & 4-2z
	\end{cases}\quad;\quad\UtilTemplate w{}cz=\begin{cases}
		z\leqslant3 & 3-z\\
		z\geqslant3 & 9-3z\text{.}
	\end{cases} \]

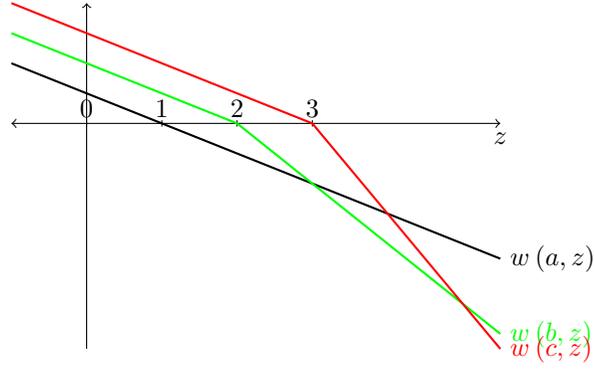
\[\begin{tikzpicture}[domain=-1:5.5,yscale=.4]
	\draw[->] (0,-7.5) -- (0,4);%
	\draw[<->] (-1,0) -- (5.5,0) node[below] {$z$};%
	\foreach \x in {0,1,2,3}
		\draw[shift={(\x,0)},color=black] (0pt,3pt) -- (0pt,-3pt) node[above] {$\x$};
	\draw[black,thick,-] (-1,2) -- (1,0) -- (5.5,-4.5) node[right]{$\UtilTemplate w{}az$};%
	\draw[green,thick,-] (-1,3) -- (2,0) -- (5.5,-7)   node[right]{$\UtilTemplate w{}bz$};%
	\draw[red,thick,-]   (-1,4) -- (3,0) -- (5.5,-7.5) node[right]{$\UtilTemplate w{}cz$};%
\end{tikzpicture}\]

Then, this preference is \POSRepresented{} by the quasi-linear utility function $\left\{ \begin{array}{l}
		\UtilTemplate u{}az=1-z\\
		\UtilTemplate u{}bz=2-z\\
		\UtilTemplate u{}cz=3-z
\end{array}\right.$ but not by $\left\{ \begin{array}{l}
	\UtilTemplate{\xPrime u}{}az=7-z\\
	\UtilTemplate{\xPrime u}{}bz=8-z\\
	\UtilTemplate{\xPrime u}{}cz=9-z
\end{array}\right.$. For instance $\left\{ \begin{array}{l}
	\UtilTemplate{\xPrime u}{}a6=1\\
	\UtilTemplate{\xPrime u}{}b6=2\\
	\UtilTemplate{\xPrime u}{}c6=3
\end{array}\right.$ but $\left\{ \begin{array}{l}
	\UtilTemplate w{}a6=-5\\
	\UtilTemplate w{}b6=-8\\
	\UtilTemplate w{}c6=-9
\end{array}\right.$.
\begin{claim}\label{claim:IfTwoPosRepThenConstantDiff}
	Let $u,w\colon\Alternatives\times\Re\rightarrow\Re$ be two quasi-linear utility functions.
\begin{enumerate}
\item If both $u$ and $w$ \POSRepresent{}  the same preference, then there exists a constant $C\in\Re$ s.t.
\[
\forall a\in\Alternatives,\,z\in\Re\quad\UtilTemplate w{}az=\UtilTemplate u{}az+C\text{.}
\]
\item If there exists a positive constant $C\geqslant0$ s.t.
\[
\forall a\in\Alternatives,\,z\in\Re\quad\UtilTemplate w{}az=\UtilTemplate u{}az-C\text{,}
\]
then $w$ \POSRepresents{}  any preference which is \POSRepresented{} by $u$.
\end{enumerate}
\end{claim}

\subsection*{The mechanism}

Given a non-empty set of alternatives $\xPrime{\Alternatives}\subseteq\Alternatives$, an agent weight vector $w\in\Delta\left(\Agents\right)$, and an alternative cost vector $c\in\Re^{\Alternatives}$, we define the following mechanism for the scenario that for all $i\in\Agents$, any type of Agent~$i$, $t\in\TypeDomI i$, is \POSRepresented{} by (at least one) quasi-linear utility function $u\in\UtilDomQL$.
\begin{mechanism}\label{MyMECHPosRep}
	\renewcommand{\SWITCH}[2]{#1}{Given a vector of reports $\left(t_{1},\ldots,t_{n}\right)$ for $t_{i}\in\protect\TypeDomI i$,\nopagebreak%
\begin{itemize}
\item[$\blacktriangleright$]  Define $u_{i}$ to be an arbitrary quasi-linear utility function which \SRepresents{} $t_{i}$.
\item[$\blacktriangleright$]
Choose $a^{\star}
	\in\protect\argmax_{a\in\protect\xPrime{\protect\Alternatives}}
	\left(\protect\AltCost a
		+\sum_{i\in\protect\Agents}
			\protect\AgentWeight i\cdot\protect\UtilTemplate uia0\right)$
and return
\begin{itemize}
\item The chosen alternative: $\protect\xOf{t_{1},\ldots,t_{n}}=a^{\star}$.
\item The payment of Agent~$i$: If $\protect\AgentWeight i>0$,
\[
\protect\piOf i{t_{1},\ldots,t_{n}}=h_{i}\left(t_{-i}\right)-\protect\sfrac 1{\protect\AgentWeight i}\cdot\left(\protect\AltCost{a^{\star}}+\sum_{j\protect\neq i}\protect\AgentWeight j\cdot\protect\UtilTemplate uj{a^{\star}}0\right)\text{,}
\]
for some function $h_{i}$ which depends only on the reports of the
other agents, and in case $\protect\AgentWeight i=0$, $\protect\piOf i{t_{1},\ldots,t_{n}}=h_{i}\left(t_{-i}\right)$.
\end{itemize}
\end{itemize}\renewcommand{\SWITCH}[2]{\undefined}}
\end{mechanism}
There is no constraint on the choice of the \POSRepresentations{} $u_{i}$ in the first step of the mechanism. Nevertheless, we note that, similarly to our note for the VCG mechanism, the choice of \POSRepresentation{} for Agent~$i$ does not influence the chosen alternative $x$ or her payment $p_{i}$. Based on Claim~\ref{claim:IfTwoPosRepThenConstantDiff}, the set $\argmax_{a\in\xPrime{\Alternatives}}\left(\AltCost a+\sum_{i\in\Agents}\AgentWeight i\cdot\UtilTemplate uia0\right)$ does not depend on the chosen \POSRepresentations{}, and the payment of Agent~$i$ is influenced only by the choice of \POSRepresentations{} for other agents. Hence, one can derive a mechanism with a more natural input language which gets (a compact representation of) some \POSRepresentation{} of the agent's preference. We show that Mechanism~\ref{MyMECHPosRep} is incentive-compatible under an additional bound on the payment functions $h_{i}$.
\begin{theorem}\label{thm:MyVCG}~\par\nopagebreak\ignorespaces%
	Let
	$\emptyset\neq\xPrime{\Alternatives}\subseteq\Alternatives$,
	$w\in\Delta\left(\Agents\right)$ be an agent weight vector,
	$c\in\Re^{\Alternatives}$ an alternative cost vector,
	and $\MECH xp$ Mechanism~\ref{MyMECHPosRep} defined by $\left\langle \xPrime A,w,c \right\rangle $.
Then $\MECH xp$ is an incentive-compatible mechanism whenever the following hold:
\begin{itemize}
\item For all $i\in\Agents$, any type of Agent~$i$ is \POSRepresented{} by a quasi-linear utility function $u\in\UtilDomQL$.
\item For all $i\in\Agents$, the function $h_{-i}$ in the definition of Mechanism~\ref{MyMECHPosRep} satisfies that for any profile of types $t_{1}\in\TypeDomI 1,\ldots,t_{n}\in\TypeDomI n$, there exist quasi-linear utility functions $u_{1},\ldots,u_{n}\in\UtilDomQL$ s.t. $u_{i}$ \POSRepresents{}  $t_{i}$ and
\[
h_{i}\left(t_{-i}\right)\leqslant\sfrac 1{\AgentWeight i}\cdot\max_{a\in\xPrime{\Alternatives}}\left[\AltCost a+\sum_{j}\AgentWeight j\cdot\UtilTemplate uja0\right]\text{.}
\]
\end{itemize}
\end{theorem}
Before proving the theorem, it is worthwhile to understand the bound on $h_{i}$ for the following two scenarios:
If all types of all agents are \Represented{} by some quasi-linear utility function,
then for any function $h_{i}$ we can find \POSRepresentations{} of the types s.t. the bound will hold (since for any constant $C>0$  $\left(u+C\right)$ \Represents{}, and hence also \POSRepresents{}, the same preference as $u$).
If all types of all agents are \POSRepresented{} by some quasi-linear utility function $u_{i}\in\UtilDomQL$ which satisfies $\min_{a\in\xPrime{\Alternatives}}\UtilTemplate uia0\geqslant0$,
then the \emph{Clarke pivot rule}~\cite[Def.~9.19]{AGTbookCh9}
	\[ h_{i}\left(t_{-i}\right)=\sfrac 1{\AgentWeight i}\cdot\max_{a\in\xPrime{\Alternatives}}\left(\AltCost a+\sum_{j\neq i}\AgentWeight j\cdot\UtilTemplate uja0\right) \]
satisfies this bound.
Last, we note that such function might not exist when the type sets are too rich. E.g., when $\AltCost{a}\equiv0$, $\AgentWeight{i}\equiv1$, and
	$\TypeDomI i = \left\{u_{\alpha}\right\}_{\alpha\in\Re}$ for $v\left(x\right)=\begin{cases} x=a &1\\ x=b &2\\ x=c &3\end{cases}$ and
	$\UtilTemplate{u}{\alpha}{x}{z}=\begin{cases}z\leqslant v\left(x\right)-\alpha & v\left(x\right)-\alpha-z \\ \OW & \left(v\left(x\right)-\alpha-z\right)^3\text, \end{cases}$
 $u_{\alpha}$ (for large enough $\alpha$) is \POSRepresented{} only by quasi-linear utility functions satisfying $\UtilTemplate{u^{QL}}{}{x}{0}\leqslant -\alpha$ for all alternatives $x\in\Alternatives$, so the bound is not feasible,
\makeatletter\@ifundefined{r@thm:MyVCG}{\def\REF{ERROR}}{\def\REF{\ref{thm:MyVCG}}}\makeatother%
\spnewtheorem*{proofOfThmVCG}{Proof of Theorem~\REF}{\itshape}{\rmfamily}%{\bfseries}{\itshape}%
\begin{proofOfThmVCG}\ThmFootnote{Our proof follows the steps of the incentive-compatibility proof of the VCG mechanism of Nisan~\cite[Prop.~9.31]{AGTbookCh9} with few modifications.}~\par\nopagebreak\ignorespaces%
	We will show that Agent~$i$ maximizes her preference by reporting truthfully. First, we notice that if $\AgentWeight i=0$ this claim is trivial since the agent has no influence on the chosen alternative or on her payment. Henceforth, we assume that $\AgentWeight i>0$.

Let $t_{i}\in\TypeDomI i$ be the preference of Agent~$i$, and let $t_{-i}\in\times_{j\neq i}\TypeDomI j$ be the reports of the other agents. Next, for $j=1,\ldots,n$, let $u_{j}$ be the quasi-linear \POSRepresentation{} of $t_{j}$ which was chosen by the mechanism.
First, we show that Agent~$i$ maximizes the utility $u_{i}$ by reporting truthfully. If an alternative~$a$ is chosen, the utility is
\[
\UtilTemplate uia0-h_{i}\left(t_{-i}\right)+\sfrac 1{\AgentWeight i}\cdot\left(\AltCost a+\sum_{j\neq i}\AgentWeight j\cdot\UtilTemplate uja0\right)\text{,}
\]
 and this expression is maximized when
\[
\AltCost a+\sum_{j}\AgentWeight j\cdot\UtilTemplate uja0
\]
is maximized which is what happens when Agent~$i$ reports truthfully.

Based on Claim~\ref{claim:IfTwoPosRepThenConstantDiff}, the set of maximizers,
	$\argmax_{a\in\xPrime{\Alternatives}}
		\left(\AltCost a+
			\sum_{j\in\Agents}\
			\AgentWeight j\cdot\UtilTemplate uja0
		\right)$,
is invariant to the choice of \POSRepresentations{} of the agents' preferences.
%
%Since $h_{i}$ satisfies
%	\[ h_{i}\left(t_{-i}\right)\leqslant\sfrac 1{\AgentWeight i}\cdot\max_{a\in\xPrime{\Alternatives}}\left[\AltCost a+\sum_{j}\AgentWeight j\cdot\UtilTemplate uja0\right]\text{,} \]
%for some choice of \POSRepresentations{}, we get that the utility $u_{i}$ (for some \POSRepresentation{}) of the outcome is positive
%
By the assumption on $h_{-i}$, there exist \POSRepresentations{} $u_{1},\ldots,u_{n}\in\UtilDomQL$ of the agents preferences s.t.
	\[
	\UtilTemplate ui{\xOf{t_{1},\ldots,t_{n}}}{\piOf i{t_{1},\ldots,t_{n}}}=
		\sfrac 1{\AgentWeight i}\cdot\left[\AltCost{\xOf{t_{1},\ldots,t_{n}}}+\sum_{j}\AgentWeight j\cdot\UtilTemplate uj{a^{\star}}0\right]-h_{i}\left(t_{-i}\right)\geqslant0\text{,}
	\]
so maximizing $u_{i}$ also maximizes the preference of Agent~$i$.\qed
\end{proofOfThmVCG}
Next, we show that when the type sets $\TypeDomI i$ are rich enough, there is no other onto incentive-compatible mechanisms.
\begin{theorem}\label{thm:MyRoberts}~\par\nopagebreak\ignorespaces%
	If there are at least three alternatives ($\sizeof{\Alternatives}\geqslant3$) and
for any $i\in\Agents$ there exist a bijection $\RobertsPhi i$ between $\TypeDomI i$ and $\left\{u\in\UtilDomQL \SetSt \min_{a\in\Alternatives}\UtilTemplate u{}a0=0\right\}$
	s.t. for any type $t\in\TypeDomI i$, $t$ is \POSRepresented{} by $\RPhiOf it$,
then any incentive-compatible onto mechanism $\MECH xp$ s.t.
	for any profile of types $t_{1}\in\TypeDomI 1,\ldots,t_{n}\in\TypeDomI n$, the payment of the \Nth{i} agent satisfies
	\[ \piOf i{t_{1},\ldots,\TypeDomI n}\leqslant\UtilTemplate {\left[\RPhiOf it\right]}{}{\xOf{t_{1},\ldots,t_{n}}}0\text, \]
can be defined as Mechanism~\ref{MyMECHPosRep} w.r.t. $\xPrime{\Alternatives}=\Alternatives$, an agent weight vector $w\in\Delta\left(\Agents\right)$, and an alternative cost vector $c\in\Re^{\Alternatives}$.
\end{theorem}
Following the steps of Roberts~\cite[Thm.~3.2]{Roberts1979}, we get as a corollary that without monetary transfers the only incentive-compatible mechanisms are dictatorships.
\begin{corollary}\label{cor:Dictatorships}~\par\nopagebreak\ignorespaces%
If there are at least three alternatives and the type sets $\left\{ \TypeDomI i\right\} _{i\in\Agents}$ satisfy the conditions of Thm.~\ref{thm:MyRoberts},
then for any incentive-compatible onto mechanism $\MECH xp$ without transfers (i.e., $\piOf i{t_{1},\ldots,t_{n}}\equiv0$ for all $i\in\Agents$), there exists a unique agent $d\in\Agents$ (a dictator) s.t. for any type profile $\left(t_{1},\ldots,t_{n}\right)$
	\[ \xOf{t_{1},\ldots,t_{n}}\in\argmax_{a\in\Alternatives}\UtilTemplate uda0\text{.} \]
\end{corollary}
\makeatletter\@ifundefined{r@thm:MyRoberts}{\def\REF{ERROR}}{\def\REF{\ref{thm:MyRoberts}}}\makeatother%
\spnewtheorem*{proofOfThmRoberts}{Proof of Theorem~\REF}{\itshape}{\rmfamily}%{\bfseries}{\itshape}%
\begin{proofOfThmRoberts}~\par\nopagebreak\ignorespaces%
	Let $\MECH xp$ be an incentive-compatible onto mechanism as stated in the theorem. We define the following auxiliary mechanism for the case in which the type sets of all agents are $\UtilDomQL$, i.e., all quasi-linear utility functions.
(Since $\RobertsPhi i$ is a bijection we use the notation $\RobertsPhi i $ for ${\RobertsPhi i}^{-1}$ as well)
\begin{mechanism*}
Given a vector of reports $\left(u_{1},\ldots,u_{n}\right)\in\left(\UtilDomQL\right)^{\Agents}$,
\begin{itemize}
	\item[$\blacktriangleright$] For $i=1,\ldots,n$ \begin{itemize}
		\item Define $\widetilde{u_{i}}\in\UtilDomQL$ by $\UtilTemplate{\widetilde{u_{i}}}{}a0=\UtilTemplate uia0-\min_{\xPrime a\in\Alternatives}\UtilTemplate ui{\xPrime a}0$.
		\item Define $t_{i}\in\TypeDomI i$ to be $\RPhiOf i{\widetilde{u}}\in\TypeDomI i$.
		\end{itemize}
	\item[$\blacktriangleright$] Apply the mechanism $\MECH xp$ on the type vector $\left(t_{1},\ldots,t_{n}\right)$ and return the same.
\end{itemize}
\end{mechanism*}
Then the following two properties hold:
\begin{itemize}
\item \uline{The mechanism is onto}.
	Let $a\in\Alternatives$. Since $\MECH xp$ is an onto mechanism, there exists a type profile $\left(t_{1}\in\TypeDomI 1,\ldots,t_{n}\in\TypeDomI n\right)$ which is mapped by the mechanism $\MECH xp$ to the alternative $a$. For $i=1,\ldots,n$, let $u_{i}$ be a quasi-linear utility function s.t. $\RPhiOf i{u_{i}}=t_{i}$. For this profile, in the first step of the mechanism $u_{i}=\widetilde{u_{i}}$ and it is mapped to $t_{i}$ and hence the chosen alternative is $a$.
\item \uline{The mechanism is incentive compatible}.
	We will show that Agent~$i$ maximizes her utility by reporting truthfully.
	Let $u_{i}\in\UtilDomQL$ be the utility of Agent~$i$, $u_{-i}\in\times_{j\neq i}\UtilDomQL$ be the reports of the other agents, and $\MECH az$ be the outcome of the mechanism when Agent~$i$ reports truthfully, i.e., for the profile $\left(u_{i},u_{-i}\right)$. Next, consider a report  for Agent~$i$, $\widehat{u_{i}}$, and let $\MECH{\widehat{a}}{\widehat{z}}$ be the outcome for the profile $\left(\widehat{u_{i}},u_{-i}\right)$.
	We define  $t_{j}\in\TypeDomI j$ for $j=1,\ldots,n$ and $\widehat{t_{i}}\in\TypeDomI i$ to be the corresponding types in the first step of the mechanism, and note that
	\[ \left\{ \begin{array}{l}
		a=\xOf{t_{i},t_{-i}}\\
		z_{i}=\piOf i{t_{i},t_{-i}}
	\end{array}\right.\quad\&\quad\left\{ \begin{array}{l}
		\widehat{a}=\xOf{\widehat{t_{i}},t_{-i}}\\
		\widehat{z}_{i}=\piOf i{\widehat{t_{i}},t_{-i}}\text{.}
	\end{array}\right. \]

	$\MECH xp$ is an incentive-compatible mechanism and hence $\MECH a{z_{i}}$ is weakly preferred over $\MECH{\widehat{a}}{\widehat{z}_{i}}$ according to the type $t_{i}$. Let $\widetilde{u_{i}}$ be the utility which $u_{i}$ is mapped to in the first step of the mechanism. Since $\widetilde{u_{i}}$ \POSRepresents{} $t_{i}=\RPhiOf i{\widetilde{u_{i}}}$ and our assumption on the payments, we get that
		$\UtilTemplate{\widetilde{u_{i}}}{}a{z_{i}}\geqslant\UtilTemplate{\widetilde{u_{i}}}{}{\widehat{a}}{\widehat{z}_{i}}$
	and hence
		\[\UtilTemplate uia{z_{i}}=
			\UtilTemplate{\widetilde{u_{i}}}{}a{z_{i}}+\min_{\xPrime a\in\Alternatives}\UtilTemplate ui{\xPrime a}0\geqslant
		\UtilTemplate{\widetilde{u_{i}}}{}{\widehat{a}}{\widehat{z}_{i}}+\min_{\xPrime a\in\Alternatives}\UtilTemplate ui{\xPrime a}0=
		\UtilTemplate ui{\widehat{a}}{\widehat{z}_{i}}\text.\]
\end{itemize}
Hence, by the theorems of Roberts~\cite[Thm.~3.1]{Roberts1979} and Green-Laffont~\cite{RePEc:ecm:emetrp:v:45:y:1977:i:2:p:427-38}, there exist an agent weight vector $w\in\Delta\left(\Agents\right)$ and an alternative cost vector $c\in\Re^{\Alternatives}$ s.t. for any report vector $\left(u_{1},\ldots,u_{n}\right)\in\left(\UtilDomQL\right)^{\Agents}$ the alternative returned by the mechanism $a^{\star}$  satisfies
	\[ a^{\star}\in\argmax_{a\in\Alternatives}\left(c_{a}+\sum_{i\in\Agents}w_{i}\cdot\UtilTemplate uia0\right)\text{,} \]
and for $i=1,\ldots,n$, the payment of Agent~$i$ equals to
	\[ \piOf i{u_{i},u_{-i}}=
	h_{i}\left(u_{-i}\right)-\sfrac 1{\AgentWeight i}\cdot\left(\sum_{j\neq i}\AgentWeight j\cdot\UtilTemplate uj{a^{\star}}0+\AltCost{a^{\star}}\right) \]
if $\AgentWeight i>0$ and to $h_{i}\left(u_{-i}\right)$ otherwise, for some function $h_{i}\colon\UtilDomI{-i}\rightarrow\Re$.

By Claim~\ref{claim:IfTwoPosRepThenConstantDiff}, if two quasi-linear utility functions $u,w\in\UtilDomQL$ \POSRepresent{} a preference $t_{i}\in\TypeDomI i$, then $u$ and $w$ differ by a constant. Since the set of maximizers,
	$\argmax_{a\in\Alternatives}\left(\AltCost a+\sum_{i\in\Agents}\AgentWeight i\cdot\UtilTemplate uia0\right)$,
is invariant to such constant shifts, $a^{\star}\in\argmax_{a\in\Alternatives}\left(\AltCost a+\sum_{i\in\Agents}\AgentWeight i\cdot\UtilTemplate uia0\right)$ for any choice of quasi-linear \POSRepresentations{} of the agents' types.
Hence, we get the required characterization of the allocation rule $x$.

Next, notice that for $u,w\in\UtilDomQL$ as above, $u$ and $w$ are mapped to the same type at the first stage of the mechanism, so also $h_{i}$ does not depend on the \POSRepresentations{}. I.e., it can be defined as a function of the types $h_{i}\colon\TypeDomI{-i}\rightarrow\Re$, which gives us the required characterization of the payment rule~$p$.\qed
\end{proofOfThmRoberts}

\subsection*{Types \Represented{} by quasi-linear utility functions}

For the special case in which all types of all agents can be \Represented{} by quasi-linear utility functions we get as corollaries the following mechanism and theorem. In particular, if the agents are reporting their quasi-linear utility functions, we get the classic incentive-compatibility characterization for quasi-linear utilities of Vickrey, Clarke, Groves, Green-Laffont, and Roberts~\cite{Vickrey1960,clarke1971,groves1973,Roberts1979,RePEc:ecm:emetrp:v:45:y:1977:i:2:p:427-38}.
\begin{mechanism}\label{MyMECHRep}
	\renewcommand{\SWITCH}[2]{#2}{Given a vector of reports $\left(t_{1},\ldots,t_{n}\right)$ for $t_{i}\in\protect\TypeDomI i$,\nopagebreak%
\begin{itemize}
\item[$\blacktriangleright$]  Define $u_{i}$ to be an arbitrary quasi-linear utility function which \SRepresents{} $t_{i}$.
\item[$\blacktriangleright$]
Choose $a^{\star}
	\in\protect\argmax_{a\in\protect\xPrime{\protect\Alternatives}}
	\left(\protect\AltCost a
		+\sum_{i\in\protect\Agents}
			\protect\AgentWeight i\cdot\protect\UtilTemplate uia0\right)$
and return
\begin{itemize}
\item The chosen alternative: $\protect\xOf{t_{1},\ldots,t_{n}}=a^{\star}$.
\item The payment of Agent~$i$: If $\protect\AgentWeight i>0$,
\[
\protect\piOf i{t_{1},\ldots,t_{n}}=h_{i}\left(t_{-i}\right)-\protect\sfrac 1{\protect\AgentWeight i}\cdot\left(\protect\AltCost{a^{\star}}+\sum_{j\protect\neq i}\protect\AgentWeight j\cdot\protect\UtilTemplate uj{a^{\star}}0\right)\text{,}
\]
for some function $h_{i}$ which depends only on the reports of the
other agents, and in case $\protect\AgentWeight i=0$, $\protect\piOf i{t_{1},\ldots,t_{n}}=h_{i}\left(t_{-i}\right)$.
\end{itemize}
\end{itemize}\renewcommand{\SWITCH}[2]{\undefined}}
\end{mechanism}

\begin{theorem}
	If all types of all agents are \Represented{} by a quasi-linear utility functions, then
\begin{itemize}
\item[$\left(a\right)$]
Let $\emptyset\neq\xPrime{\Alternatives}\subseteq\Alternatives$, $w\in\Delta\left(\Agents\right)$ be an agent weight vector, $c\in\Re^{\Alternatives}$ an alternative cost vector, and $\MECH xp$ Mechanism~\ref{MyMECHRep} defined by $\left(\xPrime A,w,c\right)$. Then $\MECH xp$ is an incentive-compatible mechanism.
\item[$\left(b\right)$]
If there are at least three alternatives ($\sizeof{\Alternatives}\geqslant3$) and for all agents $i\in\Agents$, any quasi-linear utility function $u\in\UtilDomQL$ \Represents{} some type $t\in\TypeDomI i$, then any incentive-compatible onto mechanism $\MECH xp$ can be defined as Mechanism~\ref{MyMECHRep} w.r.t. $\xPrime{\Alternatives}=\Alternatives$, an agent weight vector $w\in\Delta\left(\Agents\right)$, and an alternative cost vector $c\in\Re^{\Alternatives}$.
\end{itemize}\end{theorem}

\section*{Characterizations}

For completeness, we also show two characterizations,
	one in utility terms and one in preference terms,
of the types which are \Represented{} by a quasi-linear function
and of the types which are \POSRepresented{} by a quasi-linear function.

\begin{theorem}\label{thm:Characterizations-Util}~\par\nopagebreak\ignorespaces%
Let $u$ be a utility function which is strongly downward monotone in money. Then, \begin{itemize}
\item
$u$ is \Represented{} by a quasi-linear utility function iff there exist a function $v\colon\Alternatives\rightarrow\Re$ and a strongly monotone function $\varphi\colon\Re\rightarrow\Re$ s.t. $\UtilTemplate u{}az =\varphi\left(v\left(a\right)-z\right)$.
\item
$u$ is \POSRepresented{} by a quasi-linear utility function iff there exist a function $v\colon\Alternatives\rightarrow\Re$ and a strongly monotone function $\varphi\colon\Re\rightarrow\Re$ s.t. $\UtilTemplate u{}az=\varphi\left(v\left(a\right)-z\right)$ whenever $z\leqslant v\left(a\right)$.
\end{itemize}\end{theorem}

\begin{theorem}\label{thm:Characterizations-Type}~\par\nopagebreak\ignorespaces%
Let $\prefOver$ be a strongly downward monotone in money preference. Then, \begin{itemize}
\item
$\prefOver$ is \Represented{} by a quasi-linear utility function iff there exist an alternative $a^{\star}\in\Alternatives$ and a payment $z^{\star}\in\Re$ s.t. \begin{enumerate}
	\item
	The preference $\prefOver$ is continuous~\cite[Def.~3.C.1]{MasColell1995} in money.
	\item
	For any $x\in\Alternatives$ there exist $z_{x},\widehat{z_{x}}\in\Re$ s.t.
		$\OUTCOME x{z_{x}} \strictPrefBy\OUTCOME {a^{\star}}{z^{\star}} \strictPrefBy\OUTCOME x{\widehat{z_{x}}} $.
	\item
	For any two alternatives $x,y\in\Alternatives$ and payments $z_{x},z_{y}\in\Re$ s.t
		$\OUTCOME x{z_{x}} \prefOver\OUTCOME y{z_{y}} $
	it holds that
		\[ \forall\alpha\in\Re\qquad\OUTCOME x{z_{x}-\alpha} \prefOver\OUTCOME y{z_{y}-\alpha} \text{.} \]
\end{enumerate}
\item
$\prefOver$ is \POSRepresented{} by a quasi-linear utility function iff there exist an alternative $a^{\star}\in\Alternatives$ and a payment $z^{\star}\in\Re$ s.t. \begin{enumerate}
	\item
	For any $x\in\Alternatives$ there exist $z_{x},\widehat{z_{x}}\in\Re$ s.t. $\OUTCOME x{z_{x}} \strictPrefBy\OUTCOME {a^{\star}}{z^{\star}} \strictPrefBy\OUTCOME x{\widehat{z_{x}}} $.
	\item
	For any two alternatives $x,y\in\Alternatives$ and payments $z_{x},z_{y}\in\Re$ s.t $\left\{ \begin{array}{l}
		\OUTCOME x{z_{x}} \prefOver\OUTCOME y{z_{y}} \\
		\OUTCOME x{z_{x}} \prefOver\OUTCOME {a^{\star}}{z^{\star}}
	\end{array}\right.$
	\[ \forall\alpha\in\Re_{+}\qquad\OUTCOME x{z_{x}-\alpha} \prefOver\OUTCOME y{z_{y}-\alpha} \text{.} \]
	\end{enumerate}
\end{itemize}\end{theorem}

\section{Discussion\label{sec:Discussion}}
Identifying the preference of an agent with her cardinal utility is basic in the theory of modeling economic agents, and in most scenarios, unless additional strong assumptions are added, the utility is only a representation of the preference and does not hold any additional information of the decision behavior of the agent. This rather simple insight allows one to extend incentive-compatibility results (or similar results regarding decisions of individuals) beyond a-priori restrictive domains of utility functions. Moreover, in many cases, one's assumptions on the agents' preferences are mostly irrelevant when considering outcomes below some threshold or an outside option, which might allow us to extend such results even further.

In this work, we formalized this insight by defining a new representation notion, \POSRepresentation{}, and showed that the characterization of VCG mechanisms extends naturally and easily to type sets which are \POSRepresented{} by quasi-linear utility functions. The proof technique of a simple reduction to the works of Roberts~\cite{Roberts1979} and Green-Laffont~\cite{RePEc:ecm:emetrp:v:45:y:1977:i:2:p:427-38} also hints that one could see our characterization as a fuller version of the VCG characterization.
We expect that also the characterization of Holmstrom~\cite{Holmstrom1979} can extended in a similar fashion to characterization of type domains. Holmstrom~\cite{Holmstrom1979} shows that one can replace the assumption of unrestrictiveness of the valuation domains by a weaker assumption that the domains are smoothly-connected. We also suspect that finding the counterpart of smooth-connectiveness in preferences terms would be interesting by its own.

We expect this insight to be applicable in many more scenarios in which results deal with quasi-linear utilities while essentially not assuming the cardinal utility is an exact unique representation of the agents' preferences, and moreover when assuming it represents the preferences only for a well-defined subdomain of the outcomes.

\subsection{The quasi-linear model (The transferable utility assumption)}

As we argued in Section~\ref{sec:Introduction}, our result shows that the main driving-force of the VCG characterization is not the the ability to compare the intensity of preferences of the different agents and in particular we did not assume a transferable utility model. Our mechanism is given the ordinal preferences of the agents and outputs payments but in no place in the characterization we assumed that `money' is the same commodity for all agents and one could think of scenarios is which the designer would fine different agents using different commodities (a trivial example would be fining using different currencies). The main message we see in our work is that the main driving-force of the VCG characterization is the separability of the agents' preferences and the ability of the designer to fine the agents on a continuous scale,
by that leaving the \emph{transferable utility} domain of quasi-linear utilities and returning to the \emph{non-transferable utility} domain of Gibbard-Satterthwaite.

On the other hand, it is important to note that in our general model, $\left(a\right)$ Exactly because lacking the transferable utility assumption, \emph{budget-balancedness} becomes more complex to define and it no longer implies Pareto-optimality of the mechanism; $\left(b\right)$ A randomized mechanism is not necessarily equivalent to a deterministic one (See also below).

\section{Other ordinal generalization}

Common to the \POSRepresentation{} notion we presented here and the Parallel utilities notion of Ma et al.~\cite{Ma2018} is that both can be defined as representation by a quasi-linear utility function over a subset of the domain: In the \POSRepresentation{} case, whenever the utility is above a threshold; And in the Parallel utilities domain, whenever the utility and the payment are positive. In a subsequent work, we follow this common trait and show that our work extends easily to deal with quadrant domains in which both the utility and the payment are bounded from below by some threshold.

Moreover, we can generalize the \POSRepresentation{} assumption of the preferences to assume that a type is \POSRepresented{} by a utility function of the form $u_{i}\left(a,z\right)=u_{i}\left(a,0\right)-\phi_{i}\left(z\right)$ for some continuous monotone bijection $\phi_{i}\colon\Re\rightarrow\Re$ (i.e., an individual utility of money), and our characterization easily extends to this case as well under an additional assumption that the designer knows the functions $\phi_{i}$ for all agents.\footnote{This knowledge assumption might seems very strong, but note that this is the common assumption in the quasi-linear model.}

Last, we note that some of our initial assumptions on the mechanisms can be easily relaxed.

\subsection{Non-direct revelation mechanisms}

One could consider more general mechanisms in which the agents, instead of reporting utility functions or valuation vectors, use more abstract actions in a sequential manner or a one-round one, and instead of incentive-compatibility require implementation using dominant strategies. Applying a simple direct revelation principle~\cite{Myerson79} shows that any such general dominant strategy implementation is equivalent to a direct revelation incentive-compatible mechanism. The two mechanisms implement the same mapping of the agents' private preferences to a chosen alternative and a payment. Hence, more general mechanisms cannot implement different allocation rules than those implemented by direct-revelation mechanisms. Still, one might prefer using non-direct mechanisms for other reasons, like a more natural input language or lower complexity of the mechanism.

\subsection{Non-Deterministic mechanisms}

The characterization problem of random incentive compatible mechanisms remains open. We did not model the preferences of the agents over lotteries of outcomes, and in particular did not assume they have von Neumann-Morgenstern preferences over lotteries~\cite[Def.~6.B.5]{MasColell1995}. Hence we can not derive results regarding truthful mechanisms except the following easy claim.
\begin{corollary}
Given the conditions of Theorem~\ref{thm:MyRoberts}, the only universally truthful mechanisms~\cite[Def.~9.38]{AGTbookCh9} (i.e., truth-telling is a dominant strategy for all coin flips by the mechanism) are lotteries over mechanisms of the type of Mechanism~\ref{MyMECHPosRep}.
\end{corollary}

Extending our work to random mechanisms, even when assuming von Neumann-Morgenstern preferences, should include also extending the notion of \POSRepresentation{}. But we conjecture that in many cases such extensions, especially of \POSRepresentation{} using quasi-linear utility functions, will lose a lot of the bite of the \POSRepresentation{} notion.

\newpage{}
\bibliographystyle{splncs04}
\providecommand{\noopsort}[1]{}
\providecommand{\url}[1]{\texttt{#1}}

\newpage{}
\appendix
\section{Proof of Claim~\ref{claim:RepIsUniqueIFF}}
\begin{claim*}
Let $u,w\colon\Alternatives\times\Re\rightarrow\Re$ be two quasi-linear utility functions. Then the following three statements are equivalent:
\begin{enumerate}
\item Both $u$ and $w$ \Represent{}  the same preference over $\Alternatives\times\Re$.
\item There exists a constant $C\in\Re$ s.t. for any alternative $a\in\Alternatives$ \quad{} $\UtilTemplate u{}a0=\UtilTemplate w{}a0+C$.
\item There exists a constant $C\in\Re$ s.t. for any alternative $a\in\Alternatives$ and $z\in\Re$ \quad{} $\UtilTemplate u{}az=\UtilTemplate w{}az+C$.
\end{enumerate}
\end{claim*}

\begin{proof}~\par\nopagebreak\ignorespaces%
	\noindent{}$\underline{\left(1\right)\Rightarrow\left(2\right)}$:
If both $u$ and $w$ \Represent{}  the same preference, then
	\[ \forall a,b\in\Alternatives,\,z_{a},z_{b}\in\Re\quad\UtilTemplate u{}a{z_{a}}\geqslant\UtilTemplate u{}b{z_{b}}\iff\UtilTemplate w{}a{z_{a}}\geqslant\UtilTemplate w{}b{z_{b}}\text{.} \]
Both $u$ and $w$ are quasi-linear so we get that
	\[ \forall a,b\in\Alternatives,\,z_{a},z_{b}\in\Re\quad\UtilTemplate u{}a0-\UtilTemplate u{}b0\geqslant\left(z_{a}-z_{b}\right)\iff\UtilTemplate w{}a0-\UtilTemplate w{}b0\geqslant\left(z_{a}-z_{b}\right)\text{,} \]
and hence
	\[\forall a,b\in\Alternatives\quad\UtilTemplate u{}a0-\UtilTemplate u{}b0=\UtilTemplate w{}a0-\UtilTemplate w{}b0\]
and there exists a constant $C\in\Re$ s.t.
	\[ \forall a\in\Alternatives\quad\UtilTemplate u{}a0=\UtilTemplate w{}a0+C\text{.} \]

\noindent{}$\underline{\left(2\right)\Rightarrow\left(3\right)}$:
Since $\UtilTemplate u{}az=\UtilTemplate u{}a0-z$ and $\UtilTemplate w{}az=\UtilTemplate w{}a0-z$, this is a trivial derivation.

\noindent{}$\underline{\left(3\right)\Rightarrow\left(1\right)}$: If there exists a constant $C\in\Re$ s.t.
	\[ \forall a\in\Alternatives,\,z\in\Re\quad\UtilTemplate u{}az=\UtilTemplate w{}az+C\text{,} \]
then
	\[ \forall a,b\in\Alternatives,\,z_{a},z_{b}\in\Re\quad\UtilTemplate u{}a{z_{a}}\geqslant\UtilTemplate u{}b{z_{b}}\iff\UtilTemplate w{}a{z_{a}}\geqslant\UtilTemplate w{}b{z_{b}}\text{.}\tag*{\qed}\]
\end{proof}

\newpage{}\section{Proof of Claim~\ref{claim:IfTwoPosRepThenConstantDiff}}
\begin{claim*}
Let $u,w\colon\Alternatives\times\Re\rightarrow\Re$ be two quasi-linear utility functions.
\begin{enumerate}
\item If both $u$ and $w$ \POSRepresent{}  the same preference, then there exists a constant $C\in\Re$ s.t.
\[
\forall a\in\Alternatives,\,z\in\Re\quad\UtilTemplate w{}az=\UtilTemplate u{}az+C\text{.}
\]
\item If there exists a positive constant $C\geqslant0$ s.t.
\[
\forall a\in\Alternatives,\,z\in\Re\quad\UtilTemplate w{}az=\UtilTemplate u{}az-C\text{,}
\]
then $w$ \POSRepresents{}  any preference which is \POSRepresented{} by $u$.
\end{enumerate}
\end{claim*}

\begin{proof}~\par\nopagebreak\ignorespaces%
	\begin{enumerate}
\item If both $u$ and $v$ \POSRepresent{}  the same preference, then for any two alternatives $a,b\in\Alternatives$ and $z_{a},z_{b}\in\Re$ s.t.
	$z_{a}\leqslant\min\left(\UtilTemplate u{}a0,\UtilTemplate w{}a0\right)$
we get that $\UtilTemplate u{}a{z_{a}}\geqslant0$ and $\UtilTemplate w{}a{z_{a}}\geqslant0$ so
\[
\UtilTemplate u{}a{z_{a}}\geqslant\UtilTemplate u{}b{z_{b}}\iff\UtilTemplate w{}a{z_{a}}\geqslant\UtilTemplate w{}b{z_{b}}\text{.}
\]
Both $u$ and $w$ are quasi-linear so we get that
\[
\UtilTemplate u{}a0-\UtilTemplate u{}b0\geqslant\left(z_{a}-z_{b}\right)\iff\UtilTemplate w{}a0-\UtilTemplate w{}b0\geqslant\left(z_{a}-z_{b}\right)\text{,}
\]
and hence
\[
\forall a,b\in\Alternatives\quad\UtilTemplate u{}a0-\UtilTemplate u{}b0=\UtilTemplate w{}a0-\UtilTemplate w{}b0
\]
and there exists a constant $C\in\Re$ s.t.
\[
\begin{array}{c}
\forall a\in\Alternatives\quad\UtilTemplate u{}a0=\UtilTemplate w{}a0+C\text{, and}\\
\forall a\in\Alternatives,\,z\in\Re\quad\UtilTemplate u{}az=\UtilTemplate w{}az+C\text{.}
\end{array}
\]
\item Let $\prefOver$ be a preference which is \POSRepresented{}  by $u$.
Let $a,b\in\Alternatives$ be two alternatives and $z_{a},z_{b}\in\Re$ two payments s.t. either
	$\UtilTemplate w{}a{z_{a}}\geqslant0$ or
	$\UtilTemplate w{}b{z_{b}}\geqslant0$.
Then, either
	$\UtilTemplate u{}a{z_{a}}\geqslant\UtilTemplate w{}a{z_{a}}\geqslant0$ or
	$\UtilTemplate u{}b{z_{b}}\geqslant\UtilTemplate w{}b{z_{b}}\geqslant0$, and
\[
\UtilTemplate w{}a{z_{a}}\geqslant\UtilTemplate w{}b{z_{b}}\quad\iff\quad\UtilTemplate u{}a{z_{a}}\geqslant\UtilTemplate u{}b{z_{b}}\quad\iff\quad\MECH a{z_{a}}\prefOver\MECH b{z_{b}}\text{.}\tag*{\qed}
\]
\end{enumerate}
\end{proof}

\newpage{}\section{Proof of Corollary~\ref{cor:Dictatorships}}
\begin{corollary*}~\par\nopagebreak\ignorespaces%
If there are at least three alternatives and the type sets $\left\{ \TypeDomI i\right\} _{i\in\Agents}$ satisfy the conditions of Thm.~\ref{thm:MyRoberts},
then for any incentive-compatible onto mechanism $\MECH xp$ without transfers (i.e., $\piOf i{t_{1},\ldots,t_{n}}\equiv0$ for all $i\in\Agents$), there exists a unique agent $d\in\Agents$ (a dictator) s.t. for any type profile $\left(t_{1},\ldots,t_{n}\right)$
	\[ \xOf{t_{1},\ldots,t_{n}}\in\argmax_{a\in\Alternatives}\UtilTemplate uda0\text{.} \]
\end{corollary*}

\begin{proof}~\par\nopagebreak\ignorespaces%
By Thm.~\ref{thm:MyRoberts}, there exists an agent weight vector $w\in\Delta\left(\Agents\right)$ and an alternative cost vector $c\in\Re^{\Alternatives}$ s.t. for any report vector,
	\[ \xOf{t_{1},\ldots,t_{n}}\in\argmax_{a\in\Alternatives}\left(c_{a}+\sum_{i\in\Agents}w_{i}\cdot\UtilTemplate uia0\right)\text{.} \]
Let $d\in\Agents$ be an agent s.t. $\AgentWeight d>0$ and notice that if Agent~$d$ knows the actions of the others she can enforce any outcome in $\Alternatives$. Since there are no monetary transfers and the mechanism is incentive compatible we get that for any type profile $\left(t_{1},\ldots,t_{n}\right)$
	\[ \xOf{t_{1},\ldots,t_{n}}\in\argmax_{a\in\Alternatives}\UtilTemplate uda0\text{.} \]

Since this property cannot hold for two different agents the dictator must be unique.\qed
\end{proof}

\newpage{}\section{Proof of Theorem~\ref{thm:Characterizations-Util}}
\begin{claim}
\label{claim:Characterization-UtilREPbyQL}Let $u$ be a utility function which is strongly downward monotone in money. Then, $u$ is \Represented{}  by a quasi-linear utility function iff there exist a function $v\colon\Alternatives\rightarrow\Re$ and a strongly monotone function $\varphi\colon\Re\rightarrow\Re$ s.t.
	\[ \UtilTemplate u{}az=\varphi\left(v\left(a\right)-z\right)\text{.} \]
\end{claim}

\begin{proof}~\par\nopagebreak\ignorespaces%
	\noindent{}$\underline{\Leftarrow}$:
Trivial since $\varphi\left(v\left(a\right)-z\right)$ and $\left(v\left(a\right)-z\right)$ \Represent{}  the same preference.

\noindent{}$\underline{\Rightarrow}$:
By definition, there exists a function $v\colon\Alternatives\rightarrow\Re$ s.t. for any $x,y\in\Alternatives$ and $z_{x},z_{y}\in\Re$
	\[ \UtilTemplate u{}x{z_{x}}\geqslant \UtilTemplate u{}y{z_{y}}\quad\iff\quad v\left(x\right)-z_{x}\geqslant v\left(y\right)-z_{y}\text{,} \]
and in particular
	\[\begin{array}{l}
		\UtilTemplate u{}x{z_{x}}=\UtilTemplate u{}y{z_{y}}\quad\iff\quad v\left(x\right)-z_{x}=v\left(y\right)-z_{y}\text,\\
		\text{and }\UtilTemplate u{}x{z_{x}}=\UtilTemplate u{}y{z_{x}+v\left(y\right)-v\left(x\right)}\text{.}
	\end{array}\]

Let $a^{\star}$ be an arbitrary alternative and define a mapping
	$\varphi\colon\Re\rightarrow\Re$ by $\varphi\left(\alpha\right)=\UtilTemplate u{}{a^{\star}}{v\left(a^{\star}\right)-\alpha}$.
Then,
\begin{varwidth}[t]{\linegoal}
\uline{$\varphi$ is monotone}:
If $\alpha<\beta$ then $v\left(a^{\star}\right)-\alpha>v\left(a^{\star}\right)-\beta$ and hence
	\[\varphi\left(\alpha\right)=
		\UtilTemplate u{}{a^{\star}}{v\left(a^{\star}\right)-\alpha}<
		\UtilTemplate u{}{a^{\star}}{v\left(a^{\star}\right)-\beta}=
		\varphi\left(\beta\right)\text{.}\]

$\underline{\UtilTemplate u{}az =\varphi\left(v\left(a\right)-z\right)}$:
For any $a\in\Alternatives$ and $z\in\Re$:
	$\begin{array}[t]{rl}
	\UtilTemplate u{}az & =\UtilTemplate u{}{a^{\star}}{z+v\left(a^{\star}}-v\left(a\right)\right)\\
	 & =\varphi\left(v\left(a\right)-z\right)\text{.}\hfill\qed
	\end{array}$
\end{varwidth}
\end{proof}

\newpage{}
\begin{claim}
\label{claim:Characterization-UtilPOSREPbyQL}Let $u$ be a utility function which is strongly downward monotone in money.
Then $u$ is \POSRepresented{} by a quasi-linear utility function iff
there exist a function $v\colon\Alternatives\rightarrow\Re$ and a strongly monotone function $\varphi\colon\Re\rightarrow\Re$ s.t.
	\[ \UtilTemplate u{}az = \varphi\left(v\left(a\right)-z\right)
		\qquad\text{ whenever }z\leqslant v\left(a\right)\text{.} \]
\end{claim}

\begin{proof}~\par\nopagebreak\ignorespaces%
	\noindent{}$\underline{\Leftarrow}$:
Trivial since $\varphi\left(v\left(a\right)-z\right)$ and $\left(v\left(a\right)-z\right)$ \Represent{}  the same preference and hence \POSRepresented{}  by the same utility functions.

\noindent{}$\underline{\Rightarrow}$:
By definition, there exists a function $v\colon\Alternatives\rightarrow\Re$ s.t. for any two alternatives $x,y\in\Alternatives$ and payments $z_{x},z_{y}\in\Re$, if $v\left(x\right)\geqslant z_{x}$,
	\[ \UtilTemplate u{}x{z_{x}}\geqslant \UtilTemplate u{}y{z_{y}}\quad\iff\quad v\left(x\right)-z_{x}\geqslant v\left(y\right)-z_{y}\text{,} \]
and in particular
	\[\begin{array}{l}
		\UtilTemplate u{}x{z_{x}}=\UtilTemplate u{}y{z_{y}}\quad\iff\quad v\left(x\right)-z_{x}=v\left(y\right)-z_{y}\text,\\
		\text{and }\UtilTemplate u{}x{z_{x}}=\UtilTemplate u{}y{z_{x}+v\left(y\right)-v\left(x\right)}\text{.}
	\end{array}\]

Let $a^{\star}$ be an arbitrary alternative and define a mapping $\varphi\colon\Re\rightarrow\Re$ by
	\[\varphi\left(\alpha\right)=\begin{cases}
		\alpha\geqslant0 & \UtilTemplate u{}{a^{\star}}{v\left(a^{\star}\right)-\alpha}\\
		\alpha\leqslant0 & \UtilTemplate u{}{a^{\star}}{v\left(a^{\star}\right)}+\alpha\text{.}
	\end{cases} \]
Then, \begin{varwidth}[t]{\linegoal}

\uline{$\varphi$ is monotone}:
\begin{varwidth}[t]{\linegoal} \begin{itemize}
	\item If $\alpha<\beta\leqslant0$ then $\varphi\left(\beta\right)-\varphi\left(\alpha\right)=\beta-\alpha>0$.
	\item If $0\leqslant\alpha<\beta$ then $\varphi\left(\beta\right)-\varphi\left(\alpha\right)=\UtilTemplate u{}{a^{\star}}{v\left(a^{\star}\right)-\beta}-\UtilTemplate u{}{a^{\star}}{v\left(a^{\star}\right)-\alpha}>0$.
	\item If $\alpha\leqslant0<\beta$ then $\varphi\left(\alpha\right)=\UtilTemplate u{}{a^{\star}}{v\left(a^{\star}\right)}+\alpha\leqslant \UtilTemplate u{}{a^{\star}}{v\left(a^{\star}\right)}<\UtilTemplate u{}{a^{\star}}{v\left(a^{\star}\right)-\beta}=\varphi\left(\beta\right)$.
\end{itemize}\end{varwidth}

\uline{$\UtilTemplate u{}az=\varphi\left(v\left(a\right)-z\right)$ whenever $z\leqslant v\left(a\right)$}:
For any $a\in\Alternatives$ and $z\leqslant v\left(a\right)$:
	\[\begin{array}[t]{rl}
		\UtilTemplate u{}az & =\UtilTemplate u{}{a^{\star}}{z+v\left(a^{\star}}-v\left(a\right)\right)\\
		 & =\varphi\left(v\left(a\right)-z\right)\text{.}\hfill\qed
	\end{array}\]
\end{varwidth}
\end{proof}

\newpage{}\section{Proof of Theorem~\ref{thm:Characterizations-Type}}
\begin{claim}
\label{claim:Characterization-TypeREPbyQL}A preference $\prefOver$ is \Represented{}  by a quasi-linear utility function iff there exist an alternative $a^{\star}\in\Alternatives$ and a payment $z^{\star}\in\Re$ s.t.
	\begin{enumerate}
	\item \label{enu:REP-CONT}The preference $\prefOver$ is continuous in money.%~\cite[Def.~3.C.1]{MasColell1995}
	\item \label{enu:REP-UNBOUND}For any $x\in\Alternatives$ there exists $z_{x}\in\Re$ s.t. $\OUTCOME x{z_{x}} \strictPrefOver\OUTCOME {a^{\star}}{z^{\star}} $.

	For any $x\in\Alternatives$ there exists $z_{x}\in\Re$ s.t. $\OUTCOME x{z_{x}} \strictPrefBy\OUTCOME {a^{\star}}{z^{\star}} $.
	\item \label{enu:REP-SHIFTS}For any two alternatives $x,y\in\Alternatives$ and payments $z_{x},z_{y}\in\Re$ s.t $\OUTCOME x{z_{x}} \prefOver\OUTCOME y{z_{y}} $	it holds that
		\[ \forall\alpha\in\Re\qquad\OUTCOME x{z_{x}-\alpha} \prefOver\OUTCOME y{z_{y}-\alpha} \text{.} \]
\end{enumerate}
\end{claim}

\begin{proof}~\par\nopagebreak\ignorespaces%
	\noindent{}$\underline{\Rightarrow}$: Trivial by the properties of quasi-linear
functions.

\noindent{}$\underline{\Leftarrow}$:
First, we show that given any alternative $a\in\Alternatives$, $\prefOver$ is weakly monotone in money.
Assume for contradiction that there exist $z_{1}<z_{2}<z_{3}$ s.t.
	$\left\{ \begin{array}{l}
		\OUTCOME a{z_{1}} \strictPrefBy\OUTCOME a{z_{2}} \\
		\OUTCOME a{z_{2}} \strictPrefOver\OUTCOME a{z_{3}}
	\end{array}\right.$
and $\OUTCOME a{z_{3}} \prefBy\OUTCOME a{z_{1}} $ (the cases $\left\{ \begin{array}{l} \OUTCOME a{z_{1}} \strictPrefOver\OUTCOME a{z_{2}} \\ \OUTCOME a{z_{2}} \strictPrefBy\OUTCOME a{z_{3}} \end{array}\right.$ and $\left\{ \begin{array}{l} \OUTCOME a{z_{1}} \strictPrefBy\OUTCOME a{z_{2}} \\ \OUTCOME a{z_{2}} \strictPrefOver\OUTCOME a{z_{3}} \\ \OUTCOME a{z_{3}} \strictPrefOver\OUTCOME a{z_{1}} \end{array}\right.$ are symmetric).
By $\left(\ref{enu:REP-CONT}\right)$ there exists $\xPrime{z_{3}}\in\left(z_{1},z_{3}\right]$ s.t. $\OUTCOME a{z_{1}} \Ind\OUTCOME a{\xPrime{z_{3}}} $.
W.l.o.g., $\xPrime{z_{3}}=z_{3}$. By $\left(\ref{enu:REP-SHIFTS}\right)$, for any $0\leqslant k\leqslant n$
	\[ \OUTCOME a{z_{1}} \Ind\OUTCOME a{z_{1}+\sfrac kn\left(z_{3}-z_{1}\right)} \Ind\OUTCOME a{z_{1}} \text{.} \]

On the other hand, by $\left(\ref{enu:REP-CONT}\right)$ there exists a $\varepsilon>0$ s.t.
	\[ \forall\xPrime{z_{2}}\st\left|z_{2}-\xPrime{z_{2}}\right|<\varepsilon\qquad\OUTCOME a{\xPrime{z_{2}}} \strictPrefOver\OUTCOME a{z_{1}} \text{,} \]
so we get a contradiction.

Moreover, we show that $\prefOver$ is strongly monotone in money.
Assume for contradiction that there exist $z_{1}<z_{2}$ s.t. $\OUTCOME a{z_{1}} \Ind\OUTCOME a{z_{2}} $.
Then, by weak monotonicity, for any $z\in\left[z_{1},z_{2}\right]$
	\[ \OUTCOME a{z_{1}} \Ind\OUTCOME az \Ind\OUTCOME a{z_{2}} \text{,} \]
and by $\left(\ref{enu:REP-SHIFTS}\right)$ $\OUTCOME a{z_{1}} \Ind\OUTCOME az $ for any $z\in\Re$, in contradiction to $\left(\ref{enu:REP-UNBOUND}\right)$.

Note that it cannot be that $\prefOver$ is upward monotone in money for some $x\in\Alternatives$ and downward monotone for a different $y\in\Alternatives$.
Otherwise, by $\left(\ref{enu:REP-UNBOUND}\right)$, there exist $z_{x},z_{y}\in\Re$ s.t.
	$\OUTCOME x{z_{x}} \strictPrefOver\OUTCOME {a^{\star}}{z^{\star}} \strictPrefOver\OUTCOME y{z_{y}}$,
and by the monotonicity and $\left(\ref{enu:REP-UNBOUND}\right)$, there exists $\delta>0$ large enough s.t.
	\[ \OUTCOME x{z_{x}-\delta} \strictPrefBy\OUTCOME {a^{\star}}{z^{\star}} \text{ and }\OUTCOME y{z_{y}-\delta} \strictPrefOver\OUTCOME {a^{\star}}{z^{\star}} \text{,} \]
in contradiction to $\left(\ref{enu:REP-SHIFTS}\right)$.

W.l.o.g, assume that $\prefOver$ is strongly downward monotone in money (Otherwise, we'll show it is \Represented{}  by a utility function quasi-linear  in $\left(-z\right)$).
Next, we define $v\colon\Alternatives\rightarrow\Re$ by
	\[ v\left(x\right)=\max\left\{ z\SetSt\OUTCOME xz \prefOver\OUTCOME {a^{\star}}{z^{\star}} \right\} \text{,} \]
and note that
\begin{varwidth}[t]{\linegoal} \begin{itemize}
\item By $\left(\ref{enu:REP-CONT}\right)$ the set $\left\{ z\SetSt\OUTCOME xz \prefOver\OUTCOME {a^{\star}}{z^{\star}} \right\} $ is closed.
\item By $\left(\ref{enu:REP-UNBOUND}\right)$ the set $\left\{ z\SetSt\OUTCOME xz \prefOver\OUTCOME {a^{\star}}{z^{\star}} \right\} $ is not empty and bounded from above.
\end{itemize}\end{varwidth}

\noindent{}Hence, the maximum exists. Moreover, because of $\left(\ref{enu:REP-CONT}\right)$, it holds that $\OUTCOME x{v\left(x\right)} \Ind\OUTCOME {a^{\star}}{z^{\star}} $ and
	\[ \forall x,y\in\Alternatives\qquad\OUTCOME x{v\left(x\right)} \Ind\OUTCOME y{v\left(y\right)} \text{.} \]

Now, for any $x,y\in\Alternatives$ and $z_{x},z_{y}\in\Re$
	\[ \begin{array}{rl}
		\OUTCOME x{z_{x}} \prefOver\OUTCOME y{z_{y}}  & \iff_{\left(\ref{enu:REP-SHIFTS}\right)}\OUTCOME x{v\left(x\right)} \prefOver\OUTCOME y{z_{y}+v\left(x\right)-z_{x}} \\
		 & \iff\OUTCOME y{v\left(y\right)} \prefOver\OUTCOME y{z_{y}+v\left(x\right)-z_{x}} \\
		 & \iff_{\left(MON\right)}v\left(y\right)\leqslant z_{y}+v\left(x\right)-z_{x}\\
		 & \iff v\left(x\right)-z_{x}\geqslant v\left(y\right)-z_{y}\text{.}\hfill\qed
	\end{array} \]
\end{proof}

\newpage{}
\begin{claim}
\label{claim:Characterization-TypePOSREPbyQL}A strongly downward monotone in money preference $\prefOver$ is \POSRepresented{} by a quasi-linear utility function
iff there exist an alternative $a^{\star}\in\Alternatives$ 	and a payment $z^{\star}\in\Re$ s.t.
\begin{enumerate}
\item\label{enu:POSREP-UNBOUND}
For any $x\in\Alternatives$ there exists a payment $z_{x}\in\Re$ s.t.
	$\OUTCOME x{z_{x}}\strictPrefOver \OUTCOME {a^{\star}}{z^{\star}} $.

For any $x\in\Alternatives$ there exists a payment $z_{x}\in\Re$ s.t.
	$\OUTCOME x{z_{x}}\strictPrefBy \OUTCOME {a^{\star}}{z^{\star}} $.

\item\label{enu:POSREP-SHIFTS}
For any two alternatives $x,y\in\Alternatives$ and payments $z_{x},z_{y}\in\Re$ s.t
	$\left\{ \begin{array}{l}
		\OUTCOME x{z_{x}} \prefOver\OUTCOME y{z_{y}} \\
		\OUTCOME x{z_{x}} \prefOver\OUTCOME {a^{\star}}{z^{\star}}
	\end{array}\right.$
	\[\forall\alpha\in\Re_{+}\qquad \OUTCOME x{z_{x}-\alpha} \prefOver \OUTCOME y{z_{y}-\alpha} \text{.} \]
\end{enumerate}
\end{claim}

\begin{proof}~\par\nopagebreak\ignorespaces%
	\noindent{}$\underline{\Rightarrow}$:
Trivial by the properties of quasi-linear functions and choosing $a^{\star}$ and $z^{\star}$ s.t. the quasi-linear function satisfies $\UtilTemplate u{}{a^{\star}}{z^{\star}}=0$.

\noindent{}$\underline{\Leftarrow}$:
First, we note that since $\prefOver$ is strongly monotone, it is continuous in money.\footnote{%
		We note that this claim does not hold for weakly downward monotone preferences.
		For example, the preference which is defined by $\UtilTemplate u{}az =\begin{cases}
    	    z\leqslant0 & 1\\
        	z>0 & 0
	    \end{cases}$ is weakly downward monotone in money but is not continuous.
	    The preference is not continuous since $\forall n\quad\OUTCOME a{\sfrac 1n} \prefBy\OUTCOME a{\sfrac 12}$ but $\OUTCOME a0 \strictPrefOver \OUTCOME a{\sfrac 12} $.}
Let $a\in\Alternatives$ and $\left\{ z_{t}\right\} ,\left\{ w_{t}\right\} $ be two sequences of monetary transfers s.t.
    $\forall t\quad\OUTCOME a{z_{t}} \prefBy\OUTCOME a{w_{t}} $,
    $\lim_{t\rightarrow\infty}z_{t}=z$, and $\lim_{t\rightarrow\infty}w_{t}=p$.
Then, $\forall t\quad w_{t}\leqslant z_{t}$ and hence $w\leqslant z$ and $\OUTCOME aw \prefOver\OUTCOME az $.

Next, we define $v\colon\Alternatives\rightarrow\Re$ by
    \[ v\left(x\right)=\max\left\{ z\SetSt\OUTCOME xz \prefOver\OUTCOME {a^{\star}}{z^{\star}} \right\} \text{,} \]
and note that
    \begin{varwidth}[t]{\linegoal}\begin{itemize}
        \item By continuity the set $\left\{ z\SetSt\OUTCOME xz \prefOver\OUTCOME {a^{\star}}{z^{\star}} \right\} $ is closed.
        \item By $\left(\ref{enu:POSREP-UNBOUND}\right)$ the set $\left\{ z\SetSt\OUTCOME xz \prefOver\OUTCOME {a^{\star}}{z^{\star}} \right\} $ is not empty and bounded from above.
    \end{itemize}\end{varwidth}

\noindent{}Hence, the maximum exists. Moreover, because of the continuity, it holds that $\OUTCOME x{v\left(x\right)} \Ind \OUTCOME {a^{\star}}{z^{\star}} $ and
    \[ \forall x,y\in\Alternatives\qquad \OUTCOME x{v\left(x\right)} \Ind \OUTCOME y{v\left(y\right)} \text{.} \]

Now, for any $x,y\in\Alternatives$ and $z_{x},z_{y}\in\Re$ s.t. $v\left(x\right)-z_{x}\geqslant0$
    \[\begin{array}{rl}
        \OUTCOME x{z_{x}} \prefOver\OUTCOME y{z_{y}}
        & \iff_{\left(\ref{enu:POSREP-SHIFTS}\right)}\OUTCOME x{v\left(x\right)} \prefOver\OUTCOME y{z_{y}+v\left(x\right)-z_{x}} \\
        & \iff\OUTCOME y{v\left(y\right)} \prefOver\OUTCOME y{z_{y}+v\left(x\right)-z_{x}} \\
        & \iff_{\left(MON\right)}v\left(y\right)\leqslant z_{y}+v\left(x\right)-z_{x}\\
        & \iff v\left(x\right)-z_{x}\geqslant v\left(y\right)-z_{y}\text{.}\hfill\qed
    \end{array} \]
\end{proof}
\end{document}